\documentclass[a4paper,11pt]{article}

\usepackage{fullpage}
\usepackage[colorlinks,urlcolor=blue,citecolor=blue,linkcolor=blue]{hyperref}

\usepackage{amsmath}
\usepackage{color}
\usepackage{amsthm}
\usepackage{amsfonts}

\title{Sampling Multiple Edges Efficiently\footnote{A proceedings version of this paper is to appear in RANDOM 2021.}}

\author{%
    Talya Eden \thanks{CSAIL at MIT, \textit{talyaa01@gmail.com}. This work was supported by the National Science Foundation under Grant No.  CCF-1740751, the Eric and Wendy Schmidt Fund for Strategic Innovation, and Ben-Gurion University of the Negev.}
    \and Saleet Mossel \thanks{CSAIL at MIT, \textit{saleet@mit.edu}.}
    \and Ronitt Rubinfeld \thanks{CSAIL at MIT, \textit{ronitt@csail.mit.edu}. This work was supported by 
     the NSF TRIPODS program (awards CCF-1740751 and DMS 2022448),
NSF award CCF-2006664, and by the Fintech@CSAIL Initiative.}
}

\newcommand{\alg}[1]{\begin{figure*}[htb]
        \centering
		\fbox{
			\begin{minipage}{0.9\textwidth}{
	#1	
}
	\end{minipage}
}
\end{figure*}
}

\newcommand{\eps}{\epsilon}
\newcommand{\poly}{{\rm poly}}
\renewcommand{\th}{^\textrm{th}}
\newcommand{\om}{\overline{m}}
\newcommand{\og}{\overline{\gamma}}

\newcommand{\e}{\eps}

\newcommand{\EX}{\mathbb{E}}
\newcommand{\abs}[1]{\left| #1\right|}

\newcommand{\davg}{d_{\mathsf{avg}}}
\newcommand{\od}{\overline{d}_{\mathsf{avg}}}
\newcommand{\ox}{\overline{x}}

\def\eps{\epsilon}

\def\dmax{d_{\rm max}}

\newtheorem{thm}{Theorem}[section]

 \newtheorem{lemma}[thm]{Lemma}
 \newtheorem{claim}[thm]{Claim}
 \newtheorem{definition}[thm]{Definition}

\newcommand{\link}[2]{\hyperref[#1]{ \color{black} #2}}

\newcommand{\sets}{ \frac{n}{\tau}\cdot \frac{35 \log(6 n t/\delta)}{\epsilon^2}}

\newcommand{\setgamma}{\frac{m(S)}{\davg\cdot|S|}}
\newcommand{\setog}{\frac{m(S)}{\od\cdot |S|}}
\newcommand{\rtlog}{\frac{\log^2(n\log(1/\delta)/\delta)}{\eps}}
\newcommand{\rt}{O\left(\max\left\{ \frac{n}{\davg \cdot x}, \sqrt{\frac{n}{\od}} \right\}\cdot  \rtlog\right)}

\newcommand{\SL}{\hyperref[alg:SampleH]{\color{black} \bf Sample-Light}}
\newcommand{\SH}{\hyperref[alg:SampleL]{\color{black} \bf Sample-Heavy}}
\newcommand{\SE}{\hyperref[alg:SampleE]{\color{black} \bf Sample-Uniform-Edge}}

\newcommand{\VH}{V_{>\tau}}
\newcommand{\VL}{V_{\leq \tau}}
\newcommand{\settau}{\frac{\ox\cdot\od}{\eps}}
\newcommand{\maxox}{\sqrt{{n}/{\od}}}
\newcommand{\EL}{E_{\leq \tau}}
\newcommand{\EH}{E_{>\tau}}

\newcommand{\good}{good}

\def\nbr{\Gamma}
\newcommand{\Preprocessing}{\hyperref[alg:Preprocessing]{\color{black} \bf Preprocessing}}

\begin{document}

\maketitle

\begin{abstract}
We present a sublinear time algorithm that allows one to sample multiple edges from a distribution that is  pointwise $\epsilon$-close to the uniform distribution, in an \emph{amortized-efficient} fashion. 
We consider the adjacency list query model, where access to a graph $G$ is given via degree and neighbor queries.

The  problem of sampling a single edge in this model has been raised by Eden and Rosenbaum (SOSA 18). 
Let $n$ and $m$ denote the number of vertices and edges of $G$, respectively.
Eden and Rosenbaum provided upper and lower bounds of $\Theta^*(n/\sqrt m)$
for sampling a single edge in general graphs (where $O^*(\cdot)$ suppresses $\poly(1/\epsilon)$ and $\poly(\log n)$ dependencies). We ask whether the query complexity lower bound  for sampling a single edge can be circumvented when multiple samples are required. 
That is, can we get an improved amortized per-sample cost if we allow a  preprocessing phase? We answer in the affirmative. 

We present an algorithm that,  
if one knows the number of required samples $q$ in advance,
has an overall cost that is sublinear in $q$, namely,  $O^*(\sqrt q \cdot(n/\sqrt m))$, which is strictly preferable to $O^*(q\cdot (n/\sqrt m))$ cost resulting from $q$ invocations of the algorithm by Eden and Rosenbaum. 

Subsequent  to a preliminary version of this work, T\v{e}tek and Thorup (arXiv, preprint) proved that this bound is essentially optimal.  
\end{abstract}

\section{Introduction}
\label{sec:intro}
\sloppy
The ability to select edges uniformly at random in a large graph or network, namely \emph{edge sampling}, is  an important primitive,
interesting both from a theoretical perspective in various models of computation (e.g., ~\cite{jowhari2011tight, ahmed2013network, Aliak,ahmed2017sampling,ER18, EdenR18-approx,ERR19,AKK19, Peng20}), and from a practical perspective  
in the study of  real-world networks (e.g., ~\cite{kashtan2004efficient, Leskovec2006, wang2011understanding, cooper2014estimating,turkoglu2017edge}).
We consider the task of outputting edges from a distribution that is close to uniform;
more precisely, the output distribution on edges will be  \emph{pointwise} $\eps$-close to the uniform distribution, so that  each edge will be  returned with probability in $[\frac{1- \eps}{m}, \frac{1+ \eps}{m}]$.
Note that this is a {stronger} notion than the  more standard notion of $\eps$-close to uniform in \emph{total variation distance} (TVD).\footnote{See Section~\ref{sec:pointwise-TVD} for a detailed discussion comparing TVD-closeness to pointwise closeness.}
We consider this task in the sublinear setting, specifically, in the  adjacency list query model, where the algorithm can perform uniform  vertex queries,  as well as  degree and neighbor queries.

Three recent algorithms have been presented for this problem in the adjacency list model. The first, by Eden and Rosenbaum~\cite{ER18}, is an  $O^*(n/\sqrt{m})$ query complexity\footnote{We note that in all the mentioned algorithms the running time is asymptotically equal to the query complexity, and therefore we limit the discussion to query complexity.} algorithm that works in general graphs.\footnote{Throughout the paper $O^*(\cdot)$ is used to suppresses $\poly(\log n/\epsilon)$ dependencies.} 
This was later refined by Eden, Ron, and Rosenbaum~\cite{ERR19} to an $O^*(m\alpha/n)$ algorithm for graphs that have arboricity\footnote{The arboricity of a graph is the minimal number of forests required to cover its edge set.} at most $\alpha$ (where it is assumed that $\alpha$ is given as input to the algorithm).
Finally, in~\cite{tetek2020sampling}, T\v{e}tek and Thorup combined techniques from the previous two works and presented  the state of the art algorithm for  sampling a single 
edge. This algorithm  exponentially improves on the dependency in $1/\eps$ compared to the algorithm by~\cite{ER18}. 
All of these algorithms were also shown to be essentially optimal if one is interested in outputting a \emph{single} edge sample. 
Naively, to sample $q$ edges in general graphs, one can invoke the \cite{tetek2020sampling}
algorithm $q$ times, with expected complexity $O^*(q\cdot (n/\sqrt m))$. In this paper, we prove that this query complexity can be improved to $O^*(\sqrt q\cdot (n/\sqrt m))$. 
That is, we prove that there  exists an algorithm with a better \emph{amortized} query complexity.

\subsection{Results}
 We  present an  algorithm that returns an edge from a distribution that is pointwise $\eps$-close to uniform, and efficiently supports many edge sample invocations.
Assuming one knows in advance the number of required edge samples $q$, 
the overall cost of $q$ edge samples is $O^*(q\cdot ( n/\sqrt m)+q)=O^*(q\cdot (n/\sqrt m))$, where the equality is since we can assume that $q=O(n^2/m)$.\footnote{Observe that  if the number of required samples $q$ exceeds $n^2/m$, then one an simply perform $O(n^2\log n/m)$ uniform pair queries and with high probability recover all edges in the graph. Hence, we can assume that $q\leq n^2/m$, and so the term $q$ does not asymptotically affect the complexity.}
Subsequent to a preliminary version of this work,  T\v{e}tek and Thorup~\cite[Theorem 15 ]{tetek2020sampling}
proved that the above result is  essentially optimal.

Our algorithm is based on two procedures: a preprocessing procedure that is invoked once, and a sampling procedure which is invoked whenever an edge sample is requested. 
There is a trade-off between the preprocessing cost and per-sample cost of the sampling procedure. Namely, for a trade-off parameter $x\geq 1$, which can be given as input to the algorithm,  the preprocessing query complexity is  $O^*(n^2/(m\cdot x))$ and the per-sample cost of the sampling procedure is $O(x/\eps)$.

\begin{thm}[Informal.]\label{thm:main}
Let $G$ be a graph over $n$ vertices  and $m$ edges. 
Assume access to $G$ is given via the adjacency list query model.
There exists an algorithm that, 
given 
an approximation parameter $\eps$
 and a trade-off parameter $x$, has two procedures: a preprocessing procedure, and a sampling procedure. 
The sampling procedure
outputs an edge from a distribution that is pointwise $\eps$-close to uniform.
The  preprocessing procedure has $O^*(n^2/(m \cdot x))$ expected query complexity, and the  expected per-sample query complexity of the sampling procedure is $O(x/\eps)$.
\end{thm}

As mentioned previously, this result is essentially optimal, due to a lower bound by T\v{e}tek and Thorup. 

\begin{thm}[Theorem 15 in~\cite{tetek2020sampling}, restated]
Let $\eps$ be some small constant $0<\eps<1$.
Any  algorithm that  samples $q$ edges from a distribution  that is  pointwise $\eps$-close  to uniform  in  the adjacency list query  model must perform $\Omega(\sqrt q\cdot (n/\sqrt m))$ queries.
\end{thm}

To better understand how the complexity of our upper bound compares to what was previously known, we give some possible instantiations.
     First, setting $x=n/\sqrt{m}$ implies a preprocessing phase with  $O^*(n/\sqrt m)$ queries  and a cost of  $O(n/\sqrt m)$ per sample, thus recovering the bounds of~\cite{ER18}. 
     Second,  setting $x=1$ implies a preprocessing phase with $O(n^2/m)$ queries
and  a cost of  $O(1/\eps)$ per sample. This can be compared to the  naive approach
of querying the degrees of all the vertices in the graph, and then sampling each vertex with probability proportional to its degree and returning an edge incident to the sampled vertex.\footnote{Indeed, the naive approach returns an edge from a distribution that is \emph{exactly} uniform.}  
Hence, the naive approach yields an $O(n)$  preprocessing cost and $O(1)$ per-sample cost while our algorithm with $x=1$ yields an $O^*(n^2/m)=O^*(n/\davg)$ preprocessing and $O(1/\eps)$ per-sample cost, where $\davg$ denotes the average degree of the graph.

For a  concrete example, consider the case where $m=\Theta(n)$ and $q=O(\sqrt n)$ edge samples are required. Setting $x=n^{1/4}$  gives an overall cost of $n^{3/4}$ for sampling $q$ edges, where previously this would have required $O(n)$ queries (by either the naive approach, or performing $O(\sqrt n)$ invocations of the $O^*(n/\sqrt m)=O^*(\sqrt n)$ algorithm of \cite{ tetek2020sampling}).
In general, if the number of queries $q$ is known in advance, then setting $x=\frac{n/\sqrt m}{\sqrt q}$, yields that sampling $q$ edges has an overall cost of $O^*(\sqrt q \cdot (n/\sqrt m))$, where previously this would have required $O^*(q\cdot (n/\sqrt m))$ queries  resulting from $q$ invocations of the algorithm by  \cite{tetek2020sampling}. We discuss some more concrete applications in the following section.

\paragraph{From the augmented model to the general query model.}
Recently, it has been suggested by Aliakbarpour et al.~\cite{Aliak} to consider query models that also provide queries for uniform edge samples, and multiple algorithms have since been developed for this model, e.g., ~\cite{AKK19, Peng20, BER21,TetekTriangles}.

Currently, for ``transferring" results in models that allow  uniform edge samples  back to  models that do  not allow such queries in a black-box 
manner,\footnote{This is true for results for which pointwise-close to uniform edge samples are sufficient, as in the case in all the current sublinear results that rely on edge samples (that we know of).} 
one must either (1) pay  a multiplicative cost of $O^*(n/\sqrt m)$ per query (replacing each edge sample query in an invocation of the~\cite{ER18} algorithm for sampling edges), 
 (2) pay an additive cost of $O(n)$ (using the naive approach described above), 
or (3) pay an additive cost of $O^*(n^2/m)$ if pair queries\footnote{Pair queries return whether there is an edge between two vertices in the graph.} are allowed.\footnote{As one can sample \emph{all} edges in the graph with high probability using $O^*(n^2/m)$ uniform pair queries (by the coupon collector's argument), and then return from the set of sampled edges.} 

For example, the works by Assadi, Kapralov and Khanna~\cite{AKK19}, 
Fichtenberger, Gao and Peng~
\cite{Peng20}, and Biswas, Eden and Rubinfeld~\cite{BER21} give  algorithms that rely on edge samples  for the tasks of approximately counting and uniformly sampling
arbitrary subgraphs in sublinear time.  
Specifically, these works  assume the \emph{augmented} query model which allows for  vertex, degree, neighbor, pair as well as uniform edge samples queries. When only  vertex, degree, neighbor and pair queries (without uniform edge samples) are provided, this is referred to as the \emph{general} query model~\cite{KKR04}. 
Currently, there are no dedicated algorithms for these tasks in the general model, that does not allow edge samples.
For approximating the number  of $4$-cycles, denoted $\#C_4$, the algorithms of~\cite{AKK19,Peng20} have query complexity of $O^*(m^2/\#C_4)$. For a graph with $m=O(n)$ edges and $\#C_4=\Theta(n^{3/2})$ $4$-cycles, this results in an $O^*(\sqrt n)$ query complexity in the augmented model. 
Using our algorithm, we can set $q=O(\sqrt n)$,
and approximately count the number of $\#C_4$'s  in 
$O^*(n^{3/4})$ queries in the general query model, where previously to our results this would have cost $O(n)$ queries.
We note that this ``black-box" transformation from the augmented model to the general query model is not guaranteed to be  optimal in terms of the resulting complexity in the general model.  Indeed, dedicated algorithms for counting and sampling stars and cliques in the general model, prove  that this is not the case~\cite{GRS11,ERS19stars,eden2020approximating,eden2020faster, eden2020sampClqs,TetekTriangles}. 
Nonetheless, to the best of our knowledge, no other results are currently known for subgraphs apart from  stars or cliques, and so this approach provides the only known algorithms for arbitrary subgraph counting and sampling in the general model.

\paragraph{Pointwise vs. TVD.}\label{sec:pointwise-TVD}
A more standard measure of distance between two distributions $P$ and $Q$ is the total variation distance (TVD),  $d_{TV}(P,Q)=\frac{1}{2}\sum_{x \in \Omega}|P(x)-Q(x)|$. 
Observe that this is a strictly weaker measure. That is,
pointwise-closeness implies closeness in TVD.
Thus our algorithm immediately produce a distribution that is TVD close to uniform. 
However, being close to a distribution in TVD, does not imply pointwise-closeness.\footnote{E.g., a distribution that ignores $\eps/2$-fraction of the edges and is uniform on the rest is close in TVD to uniform, but clearly it is not pointwise close.}
Furthermore, in various settings, this weaker definition is not sufficient, as is the case in some of the applications we mentioned previously. 
For instance, 
the uniform edge samples in the algorithms 
of~\cite{AKK19,Peng20} cannot be replaced in a black-box manner by edge samples that are only guaranteed to be close to uniform in TVD. For a concrete example, consider the task of approximately counting the number of triangles.  Let $G=A\cup B$ be a graph, where $A$ is a bipartite subgraph over $(1-\eps)m$ edges, and $B$ is a clique  over $\eps m$ edges. An algorithm that returns a uniformly distributed edge in $A$ is close in TVD to uniform over the entire edge set of $G$. However, it does not allow one to correctly approximate the number of triangles in $G$, as the algorithm will never return an edge from the clique, which is where all the triangles reside.

\subsection{Technical Overview}

Sampling (almost) uniformly distributed edges is equivalent to sampling vertices with probability (almost) proportional to their degree $\frac{d(v)}{2m}$.\footnote{Since if every $v$ is sampled with probability in $(1\pm \eps)\frac{d(v)}{2m}$, performing one more uniform neighbor query from $v$ implies that each specific edge $(v,w)$ in the graph is sampled with probability in $(1\pm \eps)\cdot \frac{1}{2m}$.}  Hence, from now on we focus on the latter task.

Consider first the following naive procedure for sampling vertices with probability proportional to their degree.
Assume that $\dmax$, the maximum degree in the graph is known.  Query a vertex uniformly at random and return it with probability $\frac{ d(v)} {\dmax}$; otherwise, return fail. 
Then each vertex is sampled with probability $\frac{d(v)}{n\cdot \dmax}$.
Therefore, if we repeatedly invoke the above until a vertex is returned, then each vertex is returned  
with probability $\frac{d(v)}{2m}$, as desired. 
However, the expected number of attempts until a vertex is returned is  $O(\frac{n\cdot \dmax}{m})$ (since the overall success probability of a single attempt is $\sum_{v \in V}\frac{d(v)}{n\cdot \dmax}=\frac{2m}{n\cdot \dmax}$), which could be as high as $O(\frac{n^2}{m})$ when $\dmax=\Theta(n)$.

Our idea is to partition the graph vertices into \emph{light} and \emph{heavy}, according to some degree threshold $\tau$, that will play a similar role to that of $\dmax$ in the naive procedure above. 
Our algorithm has two procedures, a preprocessing procedure and a sampling procedure.  The preprocessing procedure is  invoked once in the beginning of the algorithm, and  the  sampling procedure  is invoked every time an edge sample is requested.
In the preprocessing procedure we construct     a data structure that will later be used  to sample heavy vertices. 
In the  sampling procedure, we repeatedly try to sample a vertex, each time either a light or a heavy with equal probability, until a vertex is returned. 
To sample light vertices, we invoke the above simple procedure with $\tau$ instead of $\dmax$. Namely, sample a uniform random vertex $v$, if $d(v)\leq \tau$, return it with probability $\frac{d(v)}{\tau}$.
To sample heavy vertices, we use the data structure constructed by  the  preprocessing procedure as will be detailed shortly.

In the preprocessing procedure, we sample a set $S$ of  $O\left(\frac{n}{\tau}\cdot\frac{\log n}{\eps^2}\right)$ vertices uniformly at random. We then construct a data structure that allows to sample edges incident\footnote{We say that an edge $(u,v)$ is incident to  $S$ if either $u$ or $v$ are in $S$.} to $S$ uniformly at random. 
It holds that with high probability for every heavy vertex $v$, its number of neighbors in $S$, denoted $d_S(v)$, is close to its expected value, $d(v)\cdot \frac{|S|}{n}$.
Also, it holds that with high probability the sum of degrees of the vertices in $S$, denoted $d(S)$,  is close to its expected value, $2m\cdot \frac{|S|}{n}$. 
Hence, to sample heavy vertices, we first sample an  edge $(u,v)$ incident to $S$ uniformly  at random (without loss of generality $u\in S$) and then we check if the second endpoint $v$ is heavy. If so, we return $v$, and otherwise we fail. 
By the previous discussion on the properties of $S$, it holds that every  heavy vertex is sampled with probability approximately $\frac{d_S(v)}{d(S)}\approx \frac{d(v)}{2m}$.

\subsection{Comparison to Previous Work}
For the sake of this discussion assume that $\epsilon$ is some small constant.
Most closely related to our work, is the algorithm of ~\cite{ER18}. Their algorithm  also works by partitioning the graph's vertices to \emph{light} and \emph{heavy} vertices according to their some degree threshold $\theta$. 
Their method of sampling light edges is identical to ours: one simply samples a vertex uniformly at random, and keeps it with probability $d(v)/\theta$. 
In our algorithm,  $\tau$ is the degree threshold for light and heavy vertices, so that $\tau$ and $\theta$ plays the same role. 
The difference between our works is in the sampling of heavy vertices.
To sample heavy vertices, the algorithm of~\cite{ER18} tries to reach heavy vertices by sampling light vertices, and then querying one of their neighbors uniformly at random.
For this approach to output heavy vertices with almost equal probability to light vertices, $\theta$ must be set to $\Omega(\sqrt m)$.
Our approach for sampling heavy vertices is different, and relies on the preprocessing phase, which later allows us to reach heavy vertices with $O(1)$ queries. 
This allows us, in a sense, to decouple the dependence of the threshold $\tau$ and the success probability of sampling light vertices.
Hence, we can allow to set the degree threshold $\tau$ to smaller values, which results in a more efficient per-sample complexity (at a cost of a preprocessing step).

The algorithm of~\cite{ERR19} also outputs a uniformly distributed single edge, however in graphs with bounded arboricity $\alpha$.
Here too the algorithm first defines light vertices, setting the threshold to $\Theta(\alpha)$.
Sampling heavy edge is then performed by starting at light vertices as before, but  taking longer random walks of length $\ell$, for $\ell$ chosen uniformly in $[\log n]$. This method was later used by T\v{e}tek~\cite{tetek2020sampling} to exponentially improve the dependence in $\eps$ of  sampling a single edge in the general setting. It is an interesting open question whether there exists an algorithm for sampling multiple edges in bounded arboricity graphs which has better complexity than the algorithm of this work.

\subsection{Further Related Work}
We note that some of the related works were already mentioned, but we list them again for the sake of completeness.

\paragraph{Sampling edges in the adjacency list model.} 
As discussed previously, the most related work to ours is that of~\cite{ER18} for sampling a single edge  from an almost uniform distribution in general graphs  in $O^*(n/\sqrt m)$ expected time. This was later refined by Eden, Rosenbaum and Ron~\cite{ERR19} to an $O^*(n\alpha/m)$ expected time algorithm in bounded arboricity graphs, where a bound $\alpha$ on the arboricity of the graph at question is also given as input to the algorithm.\footnote{Note that since for all graphs $\alpha\leq \sqrt m$, this results is always at least as good as the previous one.} 
Recently, T\v{e}tek and Thorup~\cite{tetek2020sampling} proved that the dependency in $\epsilon$ in the algorithm of~\cite{ER18} could be improved from $1/\sqrt \epsilon$ to $\log(1/\epsilon)$. 
They further proved (subsequent to our work) that given additional  access to what they refer to as hash-based neighbor  queries, there exists an algorithm for sampling multiple edges (with and without replacement) from the exactly uniform distribution in $O^*(\sqrt q\cdot (n/\sqrt m))$ time.

\paragraph{The augmented edge samples model.} 
In ~\cite{Aliak}, Aliakbarpour et al. suggested a query model which allows access to  uniform edge samples and degree queries.  
In this model they presented an algorithm for approximately counting the number of $s$-stars in expected time $O^*({m}/{\#H^{1/s}})$, where $\#H$ denotes the number of $s$-stars in the graph.
In~\cite{AKK19}, 
 Assadi, Kaparalov and Khanna
 considered the  combined power  of neighbor, degree, pair and uniform vertex and  edge samples. In this model, they
presented an algorithm that approximates the number of occurrences of any arbitrary  subgraph $H$  in a graph $G$ in expected time $O^*(m^{\rho(H)}/\#H)$, where $\rho(H)$ is the fractional edge cover\footnote{The  fractional edge cover of a graph is minimum weight assignment of weights to the graph's edges, so that the sum of weights over the edges incident to each vertex is at least $1$.} of $H$, and $\#H$ is the number of occurrences of $H$ in $G$. In the same model,  Fichtenberger, Gao, and Peng~\cite{Peng20}
simplified the above algorithm and proved the same complexity for the additional task of sampling a uniformly distributed copy of $H$. 
Recently, Biswas, Eden and Rubinfeld~\cite{BER21},
paramerterized the complexity of counting and sampling arbitrary subgraph by what they refer to as the decomposition cost of $H$, improving the above results for a large family of subgraphs $H$. In~\cite{TetekTriangles}, T\v{e}tek considers this model in the context of approximately counting triangles in the super-linear regime.

\paragraph{Sampling from networks.}
Sampling from networks is a very basic primitive that is used in a host of works for studying networks' parameters (e.g., ~\cite{kashtan2004efficient, Leskovec2006, wang2011understanding, cooper2014estimating,turkoglu2017edge}). 
Most approaches  for efficiently sampling edges from networks are 
 random walk based approaches, whose complexity is proportional to the mixing time  of the network, e.g., ~\cite{Leskovec2006,gjoka2010walking,ribeiro2010estimating, mohaisen2010measuring}.
We note that our approach cannot be directly compared with that of the random walk based ones, as the 
query models are different:
The adjacency list query model assumes access to uniform vertex queries and one can only query one neighbor at a time, while random walk based approaches usually only assume access to arbitrary seed vertices and querying a node reveals its set of neighbors. Furthermore, while in theory the mixing time of a graph can be of order  $O(n)$, in practice, social networks tend to have smaller mixing times~\cite{mohaisen2010measuring}, making random walk based approaches very efficient. Still, denoting the mixing time of the network by $t_{mix}$, such approaches require one to perform $\Omega(t_{mix})$ queries in order to obtain \emph{each} new sample, thus leaving the question of a more efficient amortized sampling procedure open.

\def\t{t}
\def\p{p}
\def\failp{\frac{\delta}{3}+\frac{1}{3^\t}}
\def\bads{\frac{\delta}{3}}

\section{Preliminaries}

Let $G=(V,E)$ be an undirected simple graph over $n$ vertices.  
We consider the adjacency list query model, which assumes 
the following set of queries:
\begin{itemize}
    \item \textbf{Uniform vertex queries:} which return a uniformly distributed vertex in $V$. 
    \item \textbf{Degree queries:} $deg(v)$, which return the degree of the queried vertex.
    \item \textbf{Neighbor queries} $nbr(v,i)$ which return the $i\th$ neighbor of $v$, if one exists and $\bot$ otherwise.
\end{itemize}
We sometimes say that we perform a ``uniform neighbor query" from some vertex $v$. This can be simply implemented by choosing an index $i\in [d(v)]$ uniformly at random, and querying $nbr(v,i)$.

Throughout the paper we consider each edge from both endpoints. That is, each edge $\{u,v\}$ is considered as two oriented edges $(u,v)$ and $(v,u)$. Abusing notation, let $E$ denote the set of all oriented edges, so that  $m = |E|=\sum_{v\in V} d(v)$ and $\davg=m/n$.
Unless stated explicitly otherwise, when we say an ``edge", we refer to oriented edges. 
  
For a vertex $v\in V$ we denote by $\nbr(v)$ the set of $v$'s neighbors. For a set $S\subseteq V$ we denote by $E(S)$ the subset of edges $(u,v)$ such that $u\in S$,  and by $m(S)$ the sum of degrees of all vertices in $S$, i.e. $m(S)=|E(S)|= \sum_{v\in S}{d(v)}$. 
For every vertex $v\in V$ and set $S\subseteq V$, we denote by $d_S(v)$ the degree of $v$ in $S$, $d_S(v)=|\nbr(v)\cap S|$.

We consider the following definition of $\eps$-pointwise close distributions:
  \begin{definition}[Definition 1.1 in ~\cite{ER18}]\label{def:pointwise}
    \label{dfn:epsilon-close}
    Let $Q$ be a fixed probability distribution on a finite set $\Omega$. We say that a probability distribution $P$ is \emph{pointwise  $\eps$-close to $Q$} if for all $x \in \Omega$,
    \[
    \abs{P(x) - Q(x)} \leq \e Q(x)\,,\quad\text{or equivalently}\quad  P(X) \in(1\pm\eps)Q(X)\,.
    \]
    If $Q = U$, the uniform distribution on $\Omega$, then we say that $P$ is \emph{pointwise  $\eps$-close to uniform}.
  
  \end{definition}

\section{Multiple Edge Sampling}

As discussed in the introduction, our algorithm consists of a preprocessing procedure that creates a data structure that enables one to sample heavy vertices, and a
sampling procedure that samples an almost uniformly distributed edge. Also recall that our procedures are parameterized by a value $x$ which allows for a trade-off between the preprocessing complexity and the per-sample complexity. Namely, allowing per-sample complexity of $O(x/\eps)$, our preprocessing procedure will run in time $O^*(n/(\davg\cdot x))$. 
If one knows the number of queries, $q$, then setting $x=\frac{n/\sqrt{m}}{\sqrt{q}}$ yields the optimal trade-off between the preprocessing and the sampling.

\subsection{Preprocessing}

In this section we present our preprocessing procedure that will later allow us to sample heavy vertices.  The procedure and its analysis are similar to the procedure Sample-degrees-typical of Eden, Ron, and Seshadhri~\cite{eden2020approximating}.

The input parameters to the procedure are $n$, the number of vertices in the graph, $x$, the trade-off parameter, $\delta$, a failure probability parameter, and $\eps$, the approximation parameter. The output is a data structure that, with probability at least $1-\delta$,  allows one to sample heavy vertices with probability (roughly) proportional to their degree.

We note that we set  $\ox = \min\{x,\maxox\}$ since for values $x=\Omega(\maxox)$ it is better to  simply use the $O^*(\sqrt{n/\davg})$ per-sample algorithm of~\cite{ER18}.
We shall make use of the following theorems.
\begin{thm}[Theorem 1.1 of \cite{GR08}, restated.] \label{thm:GR}
There exists an algorithm that,  given  query access to a graph $G$ over $n$ vertices and $m$ edges, an approximation parameter $\eps\in(0,\frac{1}{2})$, and a failure parameter $\delta\in(0,1)$, returns a value $\om$ such that with probability at least $1-\delta$, $\om \in [(1-\eps)m, m]$.
The expected query complexity and running time of the algorithm are $O(\frac{n}{\sqrt m}\cdot \frac{\log ^2n}{\eps^{2.5}})$.
\end{thm}

\begin{thm}[Section 4.2 and Lemma 17 in ~\cite{Feige06}, restated.]\label{thm:Feige}
For a set $S$ of size at least $\frac{n}{\sqrt m}\cdot \frac{34}{\eps}$, it holds that with probability at least $5/6$, 
$m(S)/s>\frac{1}{2}\cdot(1-\eps)\cdot \davg$.
\end{thm}

\begin{thm}[A data structure for a discrete distribution (e.g., \cite{walker1974new, walker1977efficient, marsaglia2004fast}).]\label{thm:DS}
There exists an algorithm that receives as input a discrete probability distribution $P$ over $\ell$ elements, and constructs a data structure that allows one to sample from $P$ in linear time $O(\ell).$
\end{thm}

\alg{
	{\bf \Preprocessing$\;(n,\eps, \delta, x)$} \label{alg:Preprocessing}
	\smallskip
	\begin{enumerate}
		\item Invoke the algorithm of~\cite{GR08}\footnote{See Theorem~\ref{thm:GR}} to get  an estimate $\od$ of the average degree $\davg$.\label{alg:step-od}
		\item Let $\ox=\min\left\{x, \maxox\right\}$\label{alg:step-ox} 
		\item Let $\t=\lceil\log_3(\frac{3}{\delta})\rceil$, and let $\tau=\settau$.\label{alg:step-tau}
		\item For $i=1$ to $\t$ do: \label{alg:find-S-loop}
		\begin{enumerate}
			\item Let $S_i$ be a multiset of $s=\sets$ vertices chosen uniformly 
			at random. 
			\item Query the degrees of all the vertices in $S_i$ and compute $m(S_i)=\sum_{v\in S_i}{d(v)}$.
		\end{enumerate}
		\item Let $S$ be the first set $S_{i}$ such that  $ \frac{m(S_{i})}{s}\in\left[\frac{1}{4}  \cdot \od , 12 \cdot \od\right]$.  \label{alg:step-choose-S}
		\begin{enumerate}
			\item	If no such set exists, then  \textbf{return} \emph{fail}. \label{step:ret-fail} 
			\item Else, set up a data structure\footnote{See Theorem~\ref{thm:DS}} $D(S)$ that supports sampling each vertex $v\in S$ with probability $\frac{d(v)}{m(S)}$. \label{alg:step-ds-for-S}
		\label{step:S}
		\end{enumerate}
		\item Let $\og=\setog$.\label{alg:step-og}
		\item \textbf{Return} $(\og, \tau, \ox, D(S))$.
	\end{enumerate}
}

The following definitions will be useful in order to prove the lemma regarding the performance of the \Preprocessing\ procedure.

\begin{definition}\label{def:good-set}
We say that a sampled set $S\subseteq V$ is $\eps$-\good\ if the following two conditions hold:
\begin{itemize}
    \item For every heavy vertex $v\in \VH$, $d_S(v)\in (1\pm \eps)|S|\cdot \frac{d(v)}{n}$.
    \item $\frac{m(S)}{s} \in \left[\frac{1}{4}\cdot \davg, 12\cdot \davg\right].$
\end{itemize}
\end{definition}

\begin{definition}\label{def:good-od}
We say that $\od$ is an $\eps$-\good\ estimate of $\davg$ if    $\od\in[(1-\eps)\davg, \davg].$
\end{definition}

\begin{lemma} \label{lem:preproc} Assume query access to a graph $G$ over $n$ vertices, $\eps\in(0,\frac{1}{2})$, $\delta\in(0,1)$, and $x\geq 1$.
The procedure \Preprocessing$(n,\eps,\delta, x)$, with probability at least $1-\delta$, returns a tuple $(\og, \tau, \ox, D(S))$ such that the following holds.
\begin{itemize}
\item $D(S)$ is a data structure that supports sampling a uniform edge in $E(S)$, for an $\eps$-\good\ set $S$, as defined in Definition \ref{def:good-set}.
\item $\ox\in[1,\maxox]$, 
$\tau=\settau$,  
and $\og=\frac{m(S)}{\od\cdot|S|}$,  
where $\od$ is an $\eps$-good estimate of $\davg$, as defined in Definition \ref{def:good-od}.
\end{itemize}
The expected query complexity and  running time of the procedure are $\rt$. 
\end{lemma}

\begin{proof}
We start by proving that with probability at least $1-\delta$ the set $S$ chosen in Step~\ref{alg:step-choose-S} is a good set.
Namely,
    that  
    (1) $\frac{m(S)}{|S|}\in\left[\frac{1}{4}  \cdot \od , 12 \cdot \od\right]$, and that
    (2) for all heavy vertices $v\in\VH$, $d_S(v) \in (1\pm\eps) s\cdot\frac{ d(v)}{n}$.
    
We start with proving the former.
By Theorem 1.1 of \cite{GR08} (see Theorem~\ref{thm:GR}), with probability at least $1-\frac{\delta}{3}$, $\od$ is an $\epsilon$-good estimate of $\davg$, that is
\begin{equation}\label{eq:GR08}
(1-\epsilon)\davg \leq\od\leq\davg.
\end{equation}
We henceforth condition on this event, and continue to prove the latter property. 
Fix an iteration $i\in[t]$. 
Observe that 
 $\EX\left[\frac{m(S_i)}{s}\right]=\davg$. 
By Markov's inequality,\footnote{Markov's inequality: if $X$ is a non-negative random variable and $a>0$,  $P(X \geq a)\leq \frac{E(X)}{a}$.} equation \eqref{eq:GR08}, and the assumption that $\eps\in(0,\frac{1}{2})$, 
$$\Pr\left[\frac{m(S_i)}{s} > 12\cdot\od\right]\leq \frac{\davg}{12\cdot\od}\leq\frac{1}{12(1-\epsilon)}\leq \frac{1}{6}.$$ 
Recall that $s =\sets$, $\tau=\settau$, and $\ox\leq\maxox$ and that we condition on $\od \geq (1-\eps)\davg$. Thus, $\tau\leq\frac{\sqrt{m}}{\epsilon}$, and $s\geq \frac{34}{\epsilon}\frac{n}{\sqrt{m}}$. 
Therefore, by Lemma 17 in ~\cite{Feige06} (see Theorem~\ref{thm:Feige}), for every $i$,  it holds that
\begin{equation}\label{eq:feige06}\Pr\left[\frac{m(S_i)}{s}\leq \frac{1}{2}\cdot \left(1-\epsilon\right)\davg\right]\leq \frac{1}{6}.\end{equation}
By equations \eqref{eq:GR08}, \eqref{eq:feige06}, and the assumption that $\eps\in(0,\frac{1}{2}),$
$$\Pr\left[\frac{m(S_i)}{s}<\frac{1}{4}\cdot\od\right]\leq\Pr\left[\frac{m(S_i)}{s}\leq \frac{1}{2}\cdot \left(1-\epsilon\right)\davg\right]\leq \frac{1}{6}$$
By the union bound, for every specific $i$,
$$\Pr\left[\frac{m(S_i)}{s} <\frac{1}{4}\cdot\od \text{\quad or \quad} \frac{m(S_i)}{s} >12\cdot\od \right] \leq \frac{1}{3}.$$
Hence, the probability that for all  the selected multisets $\{S_i\}_{i\in[t]}$,  either $\frac{m(S_i)}{s} <\frac{1}{4}\cdot\od$ or $\frac{m(S_i)}{s} >12\cdot\od$ 
is bounded by 
$\frac{1}{3^t}=\frac{\delta}{3}$ (recall $t=\lceil\log_3(\frac{3}{\delta})\rceil$).
Therefore, with probability at least $1-\frac{2\delta}{3}$, it holds that $\frac{m(S)}{s}\in\left[\frac{1}{4}  \cdot \od , 12 \cdot \od\right]$, and the procedure does not return \emph{fail} in Step~\ref{step:ret-fail}.

Next, we prove that there exists a high-degree vertex  $v\in\VH$ such that $d_S(v)\notin(1\pm \eps)s\cdot\frac{ d(v)}{n}$ with probability at most $\bads$. 
Fix an iteration $i\in[t]$, and let $S_i=\{u_1,\dots,u_s\}$ be the sampled set.
For any fixed high-degree vertex $v\in \VH$ and for some vertex $u\in V,$ let 
\begin{equation*}
\chi^v(u)=\begin{cases}
 1 &\text{$u$ is a neighbor of $v$} \\
 0 &\text{otherwise}
\end{cases} \;.
\end{equation*}
Observe that $\EX_{u\in V}\left[\chi^v(u)\right]=\frac{d(v)}{n}$, and that $d_{S_i}(v)=\sum_{j\in[s]}\chi^v(u_j)$.  Thus,
$\EX\left[d_{S_i}(v)\right]=s\cdot \frac{d(v)}{n}.$
Since the $\chi^v(u)$ variables are independent $\{0,1\}$ random variables,  
by the multiplicative Chernoff bound,\footnote{Multiplicative Chernoff bound: if $X_1,\dots, X_n$ are independent random variables taking values in $\{0, 1\}$, then for any $0\leq\delta\leq1$, $\Pr\left[\left|\sum_{i\in[n]}X_i -\mu\right|\geq \delta \mu\right]\leq 2e^{-\frac{\delta^2 \mu}{3}}$ where $\mu=\EX\left[\sum_{i\in[n]}X_i\right].$}
\begin{equation}\label{eq:dsv}
\Pr\left[\left|d_{S_i}(v)-\frac{s\cdot d(v)}{n}\right|\geq \eps\cdot \frac{s\cdot d(v)}{n}\right]\leq 2\exp\left({-\frac{\eps^2\cdot s\cdot d(v)}{3n}}\right)\leq\frac{\delta}{3nt},
\end{equation}
where the last inequality is by the assumption that  $\epsilon\in(0,\frac{1}{2})$, the setting of $s = \sets$, and since we fixed a heavy vertex $v$ so that  $d(v)\geq \tau$. 
By taking a union bound over all high-degree vertices, it holds that there exists $v\in\VH$ such that $d_{S_i}(v)\notin (1\pm \eps)\frac{s\cdot d(v)}{n}$ with probability at most $\frac{\delta}{3t}$. 

Hence, with probability at least $1-\delta$, $D(S)$ is a data structure of a good set $S$.
Moreover, by steps \ref{alg:step-ox}, \ref{alg:step-og}, and \ref{alg:step-tau} in the procedure \Preprocessing$(n, \eps, \delta, x)$ it holds that $\ox\in\left[1, \maxox\right]$,  $\og=\setog$, and  $\tau=\settau$ respectively. By equation \eqref{eq:GR08}, $\od$ is an $\eps$-good estimate for $\davg$.  

\sloppy
We now turn to analyze the complexity.
By \cite{GR08} (see Theorem~\ref{thm:GR}), the query complexity and running time of step \ref{alg:step-od} is $O\left(\frac{n}{\sqrt{m}} \cdot \frac{\log^2(n)}{\eps^{2.5}}\right).$ 
The expected query complexity and running time of the for loop  are $O(t\cdot s)=O(\frac{n}{\davg \cdot \ox}\cdot  \rtlog)$, where the equality holds by the setting of $s,t$ and since the expected value of $\od$ is $\davg$.
Step \ref{alg:step-choose-S} takes $O(t)$ time.  
By~\cite{walker1974new, walker1977efficient, marsaglia2004fast} (see Theorem~\ref{thm:DS}), the running time of step \ref{alg:step-ds-for-S} is $O(s)$.
All other steps takes $O(1)$ time. 
Hence, the expected query complexity and  running time are dominated by the for loop. By the setting of $\ox=\min\{x,\sqrt{n/\od}\}$  we have 
$O(s\cdot t) = O\left(\frac{n}{\od\cdot \ox}\cdot \rtlog\right)= \rt$  
which proves the claim. 
\end{proof}

\subsection{Sampling an edge}

In this section we present our sampling procedures.
The following definition and claim will be useful in our analysis.

\begin{definition} \label{def:heavy-light}
Let $\tau$
be a degree threshold. Let $\VL = \{v\in V \mid d(v)\leq \tau \}$, and let $\VH=V \setminus \VL$. We refer to $\VL$ and $\VH$ as the sets of \emph{light vertices} and \emph{heavy vertices}, respectively. 
Let $\EL=\{(u,v) \mid u \in \VL \}$ and $\EH=\{(u,v) \mid u \in \VH \}.$ 
\end{definition}

\begin{definition}\label{def:preprocess-succ}
If the procedure \Preprocessing$(n,\eps,\delta, x)$ returns a tuple 
$(\og, \tau, \ox, D(S))$ such that the following items of Lemma~\ref{lem:preproc} hold, then we say that this invocation is \emph{successful}.
\begin{itemize}
\item $D(S)$ is a data structure that supports sampling a uniform edge in $E(S)$, for an $\eps$-\good\ set $S$, as defined in Definition \ref{def:good-set}.
\item $\ox\in[1,\maxox]$, 
$\tau=\settau$,  
and $\og=\frac{m(S)}{\od\cdot|S|}$,  
where $\od$ is an $\eps$-good estimate of $\davg$, as defined in Definition \ref{def:good-od}.
\end{itemize}
\end{definition}

\begin{claim}\label{clm:og}
Let $\gamma=\setgamma$ and $\og=\setog$.
If $S$ is an $\eps$-good set, as in Definition \ref{def:good-set}, and $\od$ is an $\eps$-good estimate of $\davg$, as in Definition \ref{def:good-od}, then it holds that
$\og\in[1/4,12]$ and 
that $\gamma\in[(1-\eps)\og, \og]$.
\end{claim}
\begin{proof}
    By the assumption that $S$ is an $\eps$-good set, it holds that  $\frac{m(S)}{|S|}\in[\frac{1}{4}\cdot \od, 12\cdot \od].$  Therefore, $\og\in[\frac{1}{4},12].$
    By the assumption that $\od$ is an $\eps$-good estimate of $\davg$, namely
    $\od\in[(1-\eps)\davg, \davg],$ 
    it holds 
    that  $\gamma\in[(1-\eps)\og, \og].$
\end{proof}
%

%
%
\subsubsection{The sampling procedures}

We now present the two procedures for sampling light edges and heavy edges.

\alg{
	{\bf \SE$\;(\og, \tau, \ox, D(S), \eps)$} \label{alg:SampleE}
	\smallskip
	\begin{enumerate}
		\item While \textbf{True} do:\label{step:while-true}
		\begin{enumerate}
		    \item Sample uniformly at random a bit $b\gets \{0,1\}$.
		    \item If $b=0$ invoke {\bf Sample-Light$(\og, \tau)$}.
		    \item Otherwise, invoke {\bf Sample-Heavy$(\tau, D(S), \ox, \eps)$}.
		    \item If an edge $(v,u)$ was returned, then {\bf return} $(v,u)$.
		\end{enumerate}
	\end{enumerate}
}
\alg{
	{\bf \SL$\;(\og, \tau)$} \label{alg:SampleL}
	\smallskip
	\begin{enumerate}
		\item Sample a vertex $v \in V$ uniformly at random and query for its degree. \label{step:sample-uniform}
		\item If $d(v) > \tau$ \textbf{return} \emph{fail}.\label{step:check-light}
		\item Query a uniform neighbor of $v$. Let $u$ be the returned vertex.\label{step:sample-neighbor-light}
		\item \textbf{Return} $(v,u)$ with probability $\frac{d(v)}{\tau}\cdot \frac{1}{4\og}$, otherwise \textbf{return} \emph{fail}.\label{step:return-light}
	\end{enumerate}
}

\alg{
	{\bf \SH$\;(\tau, D(S), \ox, \eps)$} \label{alg:SampleH}
	\smallskip
	\begin{enumerate}
		\item Sample from the data structure $D(S)$ a vertex $v \in S$ with probability $\frac{d(v)}{m(S)}$\;.\label{step:fromS}
		\item Sample uniform neighbor of $v$. Let $u$ be the returned vertex.  \label{step:nbr-of-S}
		\item If $d(u) \leq \tau$ \textbf{return} \emph{fail}.\label{step:check-heavy}
		\item Sample uniform neighbor of $u$. Let $w$ be the returned vertex. \label{step:nbr-of-heavy} 
		\item \textbf{Return} $(u,w)$
		with probability $\eps/4\ox$, otherwise \textbf{return} \emph{fail}. \label{step:return-heavy}
	\end{enumerate}
}


Our procedure for sampling an edge \SE\ gets as input a tuple $(\og, \tau, \ox, D(S))$ which is the output of the procedure \Preprocessing.
Our guarantees on the resulting distribution of edge samples rely on the preprocessing being successful (see Definition~\ref{def:preprocess-succ}), which happens with probability at least $1-\delta$.

\begin{lemma}\label{lem:SampleL} 
Assume that \Preprocessing\ has been invoked successfully, as defined in Definition~\ref{def:preprocess-succ}.
The procedure \SL($\og,\tau$)\ returns an edge   in $\EL$ such that each edge is returned  with probability $\frac{\eps |S|}{4n \cdot \ox \cdot m(S)}$.  The query complexity and running time of the procedure are $O(1)$.
\end{lemma}

\begin{proof}
    Let $(v,u)$ be a fixed edge in $\EL$.
\begin{align*}
    \Pr[(v,u)\text{ returned}] &= \Pr[\text{ ($v$ is sampled in Step~\ref{step:sample-uniform}) and  
    ($u$ sampled in Step~\ref{step:sample-neighbor-light}) }
    \\ &\phantom{spacespa} \text{and ($(v,u)$ returned in Step~\ref{step:return-light})}]\\
    & = \frac{1}{n}\cdot \frac{1}{d(v)}\cdot \frac{d(v)}{\tau \cdot 4\og} \;.
\end{align*}
Note that by Claim~\ref{clm:og}, $1/4\og \leq 1$ and therefore, Step~\ref{step:return-light} is valid and the above holds.
Hence, by the setting of $\tau=\settau$ and  $\og=\setog$,
\[
\Pr[(v,u)\text{ is returned}]
=\frac{1}{n \cdot \tau \cdot4\og }
    =\frac{\eps \cdot |S|}{4n\cdot \ox \cdot m(S)}.
\]

    The procedure performs at most one degree query and one uniform neighbor query. All other operations take constant time. Therefore, the query complexity and running time of the procedure are constant.
\end{proof}



\begin{lemma}\label{lem:SampleH}
Assume that \Preprocessing\ has been invoked successfully, as defined in Definition~\ref{def:preprocess-succ}.
    The procedure \SH$(\tau, D(S), \ox, \epsilon)$ returns an edge  in $\EH$ such that each edge is returned with probability $\frac{(1\pm \eps)\eps|S|}{4n \cdot \ox \cdot m(S)}$. The query complexity and running time of the procedure are $O(1)$.
\end{lemma}
\begin{proof}
Let $(u,w)$ be an edge in $\EH$. 
We first compute the probability that $u$ is sampled in Step~\ref{step:nbr-of-S}. Recall, the data structure $D(S)$ supports sampling a vertex $v$ in $S$ with probability $\frac{d(v)}{m(S)}$. 
The probability that $u$ is sampled in Step~\ref{step:nbr-of-S} is equal to the probability that a vertex $v\in S$ which is a neighbor of $u$ is sampled in step \ref{step:fromS}, and $u$ is the selected neighbor of $v$ in Step \ref{step:nbr-of-S}.
Namely,  
    \[
    \Pr[u \text{ is sampled in Step~\ref{step:nbr-of-S}}] =\sum_{v \in S\cap\nbr(u)}\frac{d(v)}{m(S)}\cdot \frac{1}{d(v)}=\sum_{v \in S\cap \nbr(u)} \frac{1}{m(S)} = \frac{d_S(u)}{m(S)} \;.
    \]
By the assumption that 
\Preprocessing\ has been invoked successfully, so that 
$S$ is $\eps$-\good, and because $u \in \VH$,
       $$d_S(u)\in (1\pm\epsilon) \cdot|S|\cdot \frac{d(u)}{n}.$$
Hence, the probability that $(u,w)$ is returned by the procedure is 
    \begin{align*}
        \Pr[(u,w) \text{ is returned} ]&=
        \Pr[
        \text{ 
        ($u$  sampled in Step~\ref{step:nbr-of-S})
        and ($w$ sampled in Step~\ref{step:return-heavy})}
        \\ 
        &\phantom{spacespa} \text{ and ($(u,w)$ returned in Step~\ref{step:return-heavy})}]
        \\ 
        &=  \frac{d_S(u)}{m(S)}\cdot \frac{1}{d(u)}\cdot \frac{\eps}{4\ox} 
        \in   \frac{ (1\pm \eps)|S| \cdot \frac{d(u)}{n}\cdot \eps}{  m(S) \cdot d(u)\cdot 4\ox
        }
         =\frac{(1\pm \eps) \eps |S|}{4n \cdot \ox \cdot m(S)} \;.
    \end{align*} 
    
    
    The procedure performs one degree query and two neighbor queries, and the rest of the operations take constant time. Hence the query complexity and running time are constant.
\end{proof}

We are now ready to prove the formal version of Theorem~\ref{thm:main}.

\begin{thm} 

There exists an algorithm that gets as input query access to a graph $G$, $n$, the number of vertices in the graph, $\eps\in(0,\frac{1}{2})$,  an approximation parameter, $\delta\in(0,1)$, a failure parameter, and $x>1$,  a trade-off parameter.
The algorithm 
has a preprocessing procedure and a sampling procedure. 

The preprocessing procedure has expected query complexity $\rt$, 
 and it  succeeds with probability at least $1-\delta$.
If the preprocessing procedure succeeds, then each time the sampling procedure 
is invoked it  
returns an edge such that the distribution on returned edges is $2\eps$-point-wise close to uniform, as defined in Definition \ref{def:pointwise}. 
Each invocation of the sampling procedure has expected  $O(\ox/\eps)$ query and time complexity.
\end{thm}
\begin{proof}
By~\ref{lem:SampleH}, the procedure \Preprocessing\ procedure succeeds with probability at least $1-\delta$. Furthermore, it has expected running time and query complexity as stated.

Condition on the event that the invocation of \Preprocessing\ was successful. 
Let $P$ denote the distribution over the returned edges by the procedure \SE. By Lemma 2.3 in~\cite{ER18}, in order to prove that $P$ is pointwise $2\eps$-close to uniform, it suffices to prove that for every two edges $e, e'$ in the graph, $\frac{P(e)}{P(e')} \in (1\pm 2\eps)$.
    By Lemma~\ref{lem:SampleL}, every light edge $e$ is returned with probability $\frac{\eps\cdot |S|}{4n\cdot \ox\cdot m(S)}$. 
    By Lemma~\ref{lem:SampleH}, every heavy edge $e'$ is returned with probability $\frac{(1\pm \eps)\eps|S|}{4 n\cdot \ox \cdot m(S)}$.
Therefore, for every two edges $e,e'$ in the graph,  $\frac{P(e)}{P(e')} \in (1\pm 2\eps)$.

Next, we prove a lower bound on the success probability of a single invocation of the while loop in Step~\ref{step:while-true} in \SE.
\begin{align*}\Pr[\text{an edge is returned}]
&=\frac{1}{2}\Pr[\text{\SL\ returns an edge}] \\ & \phantom{space} +\frac{1}{2}\Pr[\text{\SH\ returns an edge}]
\\&\geq \frac{1}{2}|\EL|\cdot \frac{\eps \cdot |S|}{4n \cdot \ox \cdot m(S)} +\frac{1}{2}\cdot |\EH|\cdot \frac{(1- \eps)\eps \cdot |S|}{4n \cdot \ox\cdot m(S)}\\
& \geq \frac{1}{2}\cdot \frac{(1-\eps)\cdot \eps |S| \cdot m}{4n \cdot \ox\cdot m(S)} =\frac{(1- \eps)\eps}{8\gamma \ox} \geq \frac{\eps}{192 x} \;,
\end{align*}
where the second inequality is due to Claim~\ref{clm:og}, i.e. $\gamma\leq 12$.
Hence, the expected number of invocations until an edge is returned is $O(\ox/\eps)$.

\end{proof}

 \bibliography{all_bib_combined}

\begin{thebibliography}{10}

\bibitem{ahmed2017sampling}
Nesreen~K Ahmed, Nick Duffield, Theodore~L Willke, and Ryan~A Rossi.
\newblock On sampling from massive graph streams.
\newblock {\em Proceedings of the VLDB Endowment}, 10(11), 2017.

\bibitem{ahmed2013network}
Nesreen~K Ahmed, Jennifer Neville, and Ramana Kompella.
\newblock Network sampling: From static to streaming graphs.
\newblock {\em ACM Transactions on Knowledge Discovery from Data (TKDD)},
  8(2):1--56, 2013.

\bibitem{Aliak}
Maryam Aliakbarpour, Amartya~Shankha Biswas, Themis Gouleakis, John Peebles,
  Ronitt Rubinfeld, and Anak Yodpinyanee.
\newblock Sublinear-time algorithms for counting star subgraphs via edge
  sampling.
\newblock {\em Algorithmica}, 80(2):668--697, 2018.

\bibitem{AKK19}
Sepehr Assadi, Michael Kapralov, and Sanjeev Khanna.
\newblock A simple sublinear-time algorithm for counting arbitrary subgraphs
  via edge sampling.
\newblock In {\em Innovations in Theoretical Computer Science Conference
  {ITCS}}, volume 124 of {\em LIPIcs}, pages 6:1--6:20. Schloss Dagstuhl -
  Leibniz-Zentrum fuer Informatik, 2019.

\bibitem{BER21}
Amartya~Shankha Biswas, Talya Eden, and Ronitt Rubinfeld.
\newblock Towards a decomposition-optimal algorithm for counting and sampling
  arbitrary motifs in sublinear time.
\newblock In {\em Approximation, Randomization, and Combinatorial Optimization.
  Algorithms and Techniques, {APPROX/RANDOM} 2021, to appear}, 2021.

\bibitem{cooper2014estimating}
Colin Cooper, Tomasz Radzik, and Yiannis Siantos.
\newblock Estimating network parameters using random walks.
\newblock {\em Social Network Analysis and Mining}, 4(1):168, 2014.

\bibitem{ERR19}
Talya Eden, Dana Ron, and Will Rosenbaum.
\newblock The arboricity captures the complexity of sampling edges.
\newblock In {\em 46th International Colloquium on Automata, Languages, and
  Programming, {ICALP} 2019, July 9-12, 2019, Patras, Greece.}, pages
  52:1--52:14, 2019.
\newblock \href {https://doi.org/10.4230/LIPIcs.ICALP.2019.52}
  {\path{doi:10.4230/LIPIcs.ICALP.2019.52}}.

\bibitem{eden2020sampClqs}
Talya Eden, Dana Ron, and Will Rosenbaum.
\newblock Almost optimal bounds for sublinear-time sampling of $k$-cliques:
  Sampling cliques is harder than counting, 2020.
\newblock \href {http://arxiv.org/abs/2012.04090} {\path{arXiv:2012.04090}}.

\bibitem{ERS19stars}
Talya Eden, Dana Ron, and C~Seshadhri.
\newblock Sublinear time estimation of degree distribution moments: The
  arboricity connection.
\newblock {\em SIAM Journal on Discrete Mathematics}, 33(4):2267--2285, 2019.

\bibitem{eden2020faster}
Talya Eden, Dana Ron, and C~Seshadhri.
\newblock Faster sublinear approximation of the number of k-cliques in
  low-arboricity graphs.
\newblock In {\em Proceedings of the Fourteenth Annual ACM-SIAM Symposium on
  Discrete Algorithms}, pages 1467--1478. SIAM, 2020.

\bibitem{eden2020approximating}
Talya Eden, Dana Ron, and C~Seshadhri.
\newblock On approximating the number of k-cliques in sublinear time.
\newblock {\em SIAM Journal on Computing}, 49(4):747--771, 2020.

\bibitem{EdenR18-approx}
Talya Eden and Will Rosenbaum.
\newblock Lower bounds for approximating graph parameters via communication
  complexity.
\newblock In Eric Blais, Klaus Jansen, Jos{\'{e}} D.~P. Rolim, and David
  Steurer, editors, {\em Approximation, Randomization, and Combinatorial
  Optimization. Algorithms and Techniques, {APPROX/RANDOM} 2018, August 20-22,
  2018 - Princeton, NJ, {USA}}, volume 116 of {\em LIPIcs}, pages 11:1--11:18.
  Schloss Dagstuhl - Leibniz-Zentrum f{\"{u}}r Informatik, 2018.
\newblock \href {https://doi.org/10.4230/LIPIcs.APPROX-RANDOM.2018.11}
  {\path{doi:10.4230/LIPIcs.APPROX-RANDOM.2018.11}}.

\bibitem{ER18}
Talya Eden and Will Rosenbaum.
\newblock On sampling edges almost uniformly.
\newblock In Raimund Seidel, editor, {\em 1st Symposium on Simplicity in
  Algorithms, {SOSA} 2018, January 7-10, 2018, New Orleans, LA, {USA}},
  volume~61 of {\em {OASICS}}, pages 7:1--7:9. Schloss Dagstuhl -
  Leibniz-Zentrum f{\"{u}}r Informatik, 2018.
\newblock \href {https://doi.org/10.4230/OASIcs.SOSA.2018.7}
  {\path{doi:10.4230/OASIcs.SOSA.2018.7}}.

\bibitem{Feige06}
Uriel Feige.
\newblock On sums of independent random variables with unbounded variance and
  estimating the average degree in a graph.
\newblock {\em SIAM Journal on Computing}, 35(4):964--984, 2006.

\bibitem{Peng20}
Hendrik Fichtenberger, Mingze Gao, and Pan Peng.
\newblock Sampling arbitrary subgraphs exactly uniformly in sublinear time.
\newblock In Artur Czumaj, Anuj Dawar, and Emanuela Merelli, editors, {\em 47th
  International Colloquium on Automata, Languages, and Programming, {ICALP}
  2020, July 8-11, 2020, Saarbr{\"{u}}cken, Germany (Virtual Conference)},
  volume 168 of {\em LIPIcs}, pages 45:1--45:13. Schloss Dagstuhl -
  Leibniz-Zentrum f{\"{u}}r Informatik, 2020.
\newblock \href {https://doi.org/10.4230/LIPIcs.ICALP.2020.45}
  {\path{doi:10.4230/LIPIcs.ICALP.2020.45}}.

\bibitem{gjoka2010walking}
Minas Gjoka, Maciej Kurant, Carter~T. Butts, and Athina Markopoulou.
\newblock Walking in facebook: {A} case study of unbiased sampling of osns.
\newblock In {\em {INFOCOM} 2010. 29th {IEEE} International Conference on
  Computer Communications, Joint Conference of the {IEEE} Computer and
  Communications Societies, 15-19 March 2010, San Diego, CA, {USA}}, pages
  2498--2506. {IEEE}, 2010.
\newblock \href {https://doi.org/10.1109/INFCOM.2010.5462078}
  {\path{doi:10.1109/INFCOM.2010.5462078}}.

\bibitem{GR08}
Oded Goldreich and Dana Ron.
\newblock Approximating average parameters of graphs.
\newblock {\em Random Structures \& Algorithms}, 32(4):473--493, 2008.
\newblock \href {https://doi.org/10.1002/rsa.20203}
  {\path{doi:10.1002/rsa.20203}}.

\bibitem{GRS11}
Mira Gonen, Dana Ron, and Yuval Shavitt.
\newblock Counting stars and other small subgraphs in sublinear-time.
\newblock {\em SIAM Journal on Discrete Mathematics}, 25(3):1365--1411, 2011.

\bibitem{jowhari2011tight}
Hossein Jowhari, Mert Sa{\u{g}}lam, and G{\'a}bor Tardos.
\newblock Tight bounds for lp samplers, finding duplicates in streams, and
  related problems.
\newblock In {\em Proceedings of the thirtieth ACM SIGMOD-SIGACT-SIGART
  symposium on Principles of database systems}, pages 49--58, 2011.

\bibitem{kashtan2004efficient}
Nadav Kashtan, Shalev Itzkovitz, Ron Milo, and Uri Alon.
\newblock Efficient sampling algorithm for estimating subgraph concentrations
  and detecting network motifs.
\newblock {\em Bioinformatics}, 20(11):1746--1758, 2004.

\bibitem{KKR04}
Tali Kaufman, Michael Krivelevich, and Dana Ron.
\newblock Tight bounds for testing bipartiteness in general graphs.
\newblock {\em SIAM Journal on Computing}, 33(6):1441--1483, 2004.
\newblock \href {https://doi.org/10.1137/S0097539703436424}
  {\path{doi:10.1137/S0097539703436424}}.

\bibitem{Leskovec2006}
Jure Leskovec and Christos Faloutsos.
\newblock Sampling from large graphs.
\newblock In {\em Proceedings of the 12th ACM SIGKDD International Conference
  on Knowledge Discovery and Data Mining}, KDD '06, pages 631--636, New York,
  NY, USA, 2006. ACM.
\newblock URL: \url{http://doi.acm.org/10.1145/1150402.1150479}, \href
  {https://doi.org/10.1145/1150402.1150479}
  {\path{doi:10.1145/1150402.1150479}}.

\bibitem{marsaglia2004fast}
George Marsaglia, Wai~Wan Tsang, Jingbo Wang, et~al.
\newblock Fast generation of discrete random variables.
\newblock {\em Journal of Statistical Software}, 11(3):1--11, 2004.

\bibitem{mohaisen2010measuring}
Abedelaziz Mohaisen, Aaram Yun, and Yongdae Kim.
\newblock Measuring the mixing time of social graphs.
\newblock In {\em Proceedings of the 10th ACM SIGCOMM conference on Internet
  measurement}, pages 383--389, 2010.

\bibitem{ribeiro2010estimating}
Bruno Ribeiro and Don Towsley.
\newblock Estimating and sampling graphs with multidimensional random walks.
\newblock In {\em Proceedings of the 10th ACM SIGCOMM conference on Internet
  measurement}, pages 390--403, 2010.

\bibitem{tetek2020sampling}
Jakub T{\v{e}}tek and Mikkel Thorup.
\newblock Sampling and counting edges via vertex accesses.
\newblock {\em arXiv preprint arXiv:2107.03821}, 2021.

\bibitem{turkoglu2017edge}
Duru T{\"u}rkoglu and Ata Turk.
\newblock Edge-based wedge sampling to estimate triangle counts in very large
  graphs.
\newblock In {\em 2017 IEEE International Conference on Data Mining (ICDM)},
  pages 455--464. IEEE, 2017.

\bibitem{TetekTriangles}
Jakub T\v{e}tek.
\newblock Approximate triangle counting via sampling and fast matrix
  multiplication.
\newblock {\em CoRR}, abs/2104.08501, 2021.
\newblock URL: \url{https://arxiv.org/abs/2104.08501}, \href
  {http://arxiv.org/abs/2104.08501} {\path{arXiv:2104.08501}}.

\bibitem{walker1974new}
Alastair~J. Walker.
\newblock New fast method for generating discrete random numbers with arbitrary
  frequency distributions.
\newblock {\em Electronics Letters}, 10(8):127--128, 1974.

\bibitem{walker1977efficient}
Alastair~J. Walker.
\newblock An efficient method for generating discrete random variables with
  general distributions.
\newblock {\em ACM Transactions on Mathematical Software}, 3(3):253--256, 1977.

\bibitem{wang2011understanding}
Tianyi Wang, Yang Chen, Zengbin Zhang, Tianyin Xu, Long Jin, Pan Hui, Beixing
  Deng, and Xing Li.
\newblock Understanding graph sampling algorithms for social network analysis.
\newblock In {\em 2011 31st international conference on distributed computing
  systems workshops}, pages 123--128. IEEE, 2011.

\end{thebibliography}

\end{document}